\title{PDE Approaches to Graph Analysis}
\author{Survey by Justin Solomon\footnote{In satisfaction of qualifying exam requirements for computer science theory, Stanford University, fall 2012.}}
\date{October 2012 (posted April 2015)}

\newcommand{\G}[0]{\mathcal{G}}
\newcommand{\V}[0]{\mathcal{V}}
\newcommand{\E}[0]{\mathcal{E}}

\newcommand{\calc}[0]{\textrm{calc}}
\newcommand{\degree}[0]{\textrm{deg }}
\newcommand{\LL}[0]{\mathscr{L}}
\newcommand{\GG}[0]{\mathscr{G}}
\newcommand{\HH}[0]{\mathscr{H}}

\documentclass[11pt]{article}

\usepackage{mathrsfs}
\usepackage{mathpazo}
\usepackage{fullpage}
\usepackage{amsmath}
\usepackage{amsfonts}
\usepackage{amsthm}
\usepackage{xcolor}
\usepackage{units}
\usepackage{mathtools}
\usepackage{hyperref}

\hypersetup{
    colorlinks=true,
linkcolor={blue},
citecolor={blue},
urlcolor={blue},
    pdfborder={0 0 0},
}

\newcommand{\R}[0]{\mathbb{R}}
\newcommand{\Z}[0]{\mathbb{Z}}
\renewcommand{\L}[0]{\mathcal{L}}
\newcommand{\D}[0]{\mathcal{D}}

\newtheorem{lemma}{Lemma}

\begin{document}
\maketitle

The analysis, understanding, and comparison of network structures is a prominent topic not only in computer science theory but also in a diverse set of application-oriented fields.  For instance, social media sites use tools from this domain to understand large-scale structures arising from localized connections between users, while computer graphics software applies the same methodology to understanding the shapes of meshed surfaces.  Further applications of this toolkit range from designing random walk models using Markov chains~\cite{lovasz93} to characterizing patterns in molecules, neural networks, and food webs~\cite{chung07}.

The obvious structure for representing a network is a \emph{graph}, given by a collection of nodes connected by edges (optionally associated with distances or weights).  The analysis of graphs as discrete objects is a classical topic predating the development of computer technology.  Algorithms for finding shortest paths, routing flow, finding spanning subtrees, identifying cycles, and so on are well-understood, although many simple problems about graphs remain unsolved.

Some more modern branches of graph theory, however, treat graphs themselves as mutable objects that should be understood in a continuous or probabilistic context.  Graphs from many sources are subject to discrete and continuous changes:  social networks grow and evolve as users ``friend'' and ``un-friend'' each other, links can lengthen or strengthen, and so on.  Due to this mutability as well as the possibility of error in the collection of graph data, multiscale methods characterizing graphs at different levels of detail are helpful for understanding broad or changing patterns in connectivity.  Parameters such as scale and sensitivity, however, can be understood as \emph{continuously}-varying, and thus research carrying out such ``semi-discrete'' analysis with rigor must combine insights from discrete graph theory and real analysis.

Of special interest are methods making use of partial differential equations (PDEs) and spectral analysis.  PDE methodologies model how physical interactions on a local scale can characterize global behavior of waves, heat, and other phenomena on assorted media and domains.  A fundamental and predictable theme in PDE analysis is that the operators used to construct PDEs encode information about the structure of the domains with which they are associated.  For example, the differential building blocks of the heat equation determine how a given surface conducts heat.  Less obviously, the De Rham cohomology from differential topology demonstrates that the topology of a smooth surface can be understood by analyzing applications of its ``first derivative'' operator $d$~\cite{madsen97}.

While first-order PDEs largely can be solved exactly, understanding the second-order Laplacian operator leads to techniques for analyzing functions on a surface and the geometry of the surface itself~\cite{rosenberg97}.  Applications in geometry processing discretize these equations using finite elements, discrete exterior calculus (DEC), and other approximations of the Laplacian for per-vertex signals on meshed surfaces~\cite{hirani03,reuter09}.  Many such discretizations exhibit convergence in the limit of mesh refinement, providing a connection to the theory of smooth surfaces, and in the case of DEC discrete theorems provide exact descriptions of behavior without any limiting process.  Regardless, with a discrete surface Laplacian in place, tools for matching, segmentation, and other tasks can be designed using solutions to the heat and wave equations, among other model PDEs~\cite{bronstein11}.

PDE-based analysis of discrete surfaces has the advantage that it can draw intuition from parallel constructions in differential geometry.  Differential graph analysis, on the other hand, requires the development of a new framework to understand diffusion, oscillation, and other phenomena as they might occur on graphs, which may not be physically realizable.  Instead, we \emph{define} a Laplacian for functions on a graph with particular manipulations, constructions, and applications in mind.  Of particular success has been the field of ``spectral graph theory,'' which involves studying eigenvalues of matrix-based graph Laplacians and connecting their properties to those of the graph itself~\cite{chung97}.

Spectral graph theory still is a fairly discrete area of graph analysis in that Laplacians in this domain are sparse matrices and their resulting spectra are finite.  In contrast, solving PDEs like the heat and wave equations on a graph involve \emph{continuous} derivatives in a time variable.  In this nascent area of research, there is less consensus on the choice of appropriate ``differential'' operators and treatments of the time variable to yield meaningful solutions.  Such solutions, however, can yield insight into how signals might propagate from node to node of a graph through its edges.

Here, we will discuss and contrast approaches proposed for graph adaptations of three model PDEs:  the Poisson, wave, and heat equations (introduced in~\cite{chung00,friedman04,chung07}, resp.).  In particular, we will show how the choice of a graph Laplacian operator complements the solutions and properties of each model equation.  Such a direct comparison motivates discussion of modeling more complex PDEs on graphs and potential directions for future research.

\section{Basic PDE Analysis}\label{basicanalysis}

Here we introduce the model PDEs we eventually will consider on graphs.  We focus on the definition and structure of each of these equations rather than its derivation from physical principles; we refer the reader to any basic PDE textbook, such as~\cite{strauss08}, for a physically-motivated discussion.

In general, we will consider a \emph{partial differential equation} in a scalar function $u(x)$ to be any equation of the form $F(\{D^nu: n\in\Z^+\},x)=G(x)$, where $D^nu$ denotes the set of $n$th derivatives of $u$.  To simplify matters, we will consider only the case where $F$ is linear in $u$ and its derivatives, thus denoting a \emph{linear} PDE.  We take $u$ to be a sufficiently differentiable function $u:\Omega\rightarrow\R$ for some domain $\Omega$; optionally we can include an additional ``time'' variable $t\in\R^+$ in which case $u:\Omega\times\R^+\rightarrow\R$.  We remain purposefully vague in this section with respect to the choice of $\Omega$; of course one can think of $\Omega\subseteq\R^n$ for the usual discussion of model PDEs, but we eventually will take $\Omega$ to be a graph.

\subsection{The Laplacian}\label{laplacianproperties}

The ``method of characteristics'' readily provides solutions to most model first-order PDEs in $u$, but including a second-order term is sufficient to yield diverse and nontrivial behavior.  To this end, we introduce a linear \emph{Laplacian} operator $\Delta$, which evaluates the total second derivative of $u$; for instance, the Laplacian of $u:\Omega\subseteq\R^n\rightarrow\R$ takes the following form:
\begin{equation}
\Delta u = -\sum_{i=1}^n \frac{\partial^2 u}{\partial x_i^2}
\end{equation}
Symbolically we can factor the Laplacian operator from the right hand side, and thus we consider $\Delta$ as a functional $\Delta: C^\infty(\Omega)\rightarrow C^\infty(\Omega)$.  In the case of $\Omega\subseteq \R^n$ above we thus write
\begin{equation}
\Delta = -\sum_{i=1}^n \frac{\partial^2}{\partial x_i^2}.
\end{equation}

In functional analysis, the Laplacian is the canonical example of a compact self-adjoint linear operator.  Specifically, the operator satisfies the following properties:
\begin{description}
\item[Linear] For $u,v\in C^\infty(\Omega)$ and $c_1, c_2\in\R$, $\Delta(c_1u+c_2v)=c_1\Delta u+c_2\Delta v$.
\item[Compact] $\Delta$ is bounded, in that for appropriate functional norm $\|\cdot\|$ there exists some $M\geq0$ satisfying $\|\Delta u\|\leq M\|u\|$ for all $u\in C^\infty(\Omega)$.  Furthermore, $\Delta$ applied to a sufficiently small open neighborhood of the function $u\equiv0$ using the topology induced by the norm $\|\cdot\|$  yields a subset of a compact set in the same functional space.
\item[Self-adjoint] If we take $u,v\in C^\infty(\Omega)$, then $\langle \Delta u, v\rangle=\langle u,\Delta v\rangle$ for the functional inner product $\langle \cdot,\cdot\rangle$.
\end{description}
Many of these properties are difficult to verify for functional operators on $\Omega\subseteq\R^n$ but become trivial for operators such as the graph Laplacian, which operate on finite-dimensional spaces.  In this case, self-adjoint operators are given by symmetric matrices $M\in\R^{n\times n}$ (thus satisfying $M=M^\top$) since we can write $(M\vec{v})\cdot\vec{w}=(M\vec{v})^\top\vec{w}=\vec{v}^\top(M^\top \vec{w})=\vec{v}\cdot (M^\top \vec{w})$.

To distinguish the Laplacian from other elliptic operators, we will give it one more property that will make it straightforward to generalize proofs of certain analytical theorems:
\begin{description}
\item[Categorizes extrema] If $x\in\Omega$ is a local minimum of $u:\Omega\rightarrow\R$, then $[\Delta u](x)\leq0$.
\end{description}
Here, we will assume that the idea of a ``local minimum'' is adapted to the domain $\Omega$ in question.  For instance, if $\Omega\subseteq\R^n$ and $u\in C^\infty(\Omega)$, then $x$ will be a local minimum when $u(x)\leq u(y)$ for all $y$ satisfying $\|x-y\|<\varepsilon$ for some $\varepsilon>0$.  On a graph, a local minimum of a function on vertices will be a vertex whose associated value is less than or equal to the values at its neighbors.

One principal reason that self-adjoint linear operators are so appealing is that they admit a full set of orthogonal eigenvectors spanning $L^2(\Omega)$.\footnote{This discussion is again cavalier with respect to the exact properties of $\Omega$ for such spectral theory to hold.  Since we principally consider graphs and compact subsets of $\R^n$, we will not delve into the particulars of this theory.}  In particular, there exist orthonormal eigenfunctions $\phi_1, \phi_2, \ldots$ with corresponding eigenvalues $0\leq\lambda_1\leq\lambda_2\leq\ldots$ such that $f\in L^2(\Omega)$ can be written $f=\sum_ia_i\phi_i$.

For such a decomposition to exist, in many cases we must introduce boundary conditions.  Thus, if desired or necessary we decompose $\Omega$ into a disjoint boundary and interior as $\Omega=\partial\Omega\cup\mathring{\Omega}$.  In this case, we will consider only eigenfunctions $\phi_i$ satisfying $\Delta \phi_i = \lambda_i\phi_i$ in $\mathring{\Omega}$ and $\phi_i|_{\partial\Omega}=0$.  This restriction corresponds to \emph{Dirichlet} boundary conditions; we will not consider derivative-based \emph{Neumann} boundary conditions or others here, although most of the frameworks we discuss will behave equally well in this case.

\subsection{The Elliptic Case:  Laplace and Poisson Equations}

The simplest possible PDE involving the Laplacian is the Laplace equation, given by:
\begin{equation}
\Delta u = 0
\end{equation}
This PDE is slightly generalized by the Poisson equation, which prescribes nonzero values for $\Delta u$:
\begin{equation}\label{poissoneqn}
\Delta u = f
\end{equation}
The Poisson and Laplace equations are \emph{elliptic}, since the highest-order derivatives appear in the elliptic operator $\Delta$.

These equations appear in many situations, from Maxwell's equations for electrostatics and the governing equations of Brownian motion to byproducts of the Cauchy-Riemann equations defining complex holomorphic functions.  In discrete form, they are used in image smoothing because they are critical points of the functional $\int \|\nabla I\|^2\ dA$ for image $I$.

Note that if we write $f=\sum_ia_i\phi_i$ for Laplace eigenfunctions $\phi_i$ and similarly decompose $u=\sum_ib_i\phi_i$, then~\eqref{poissoneqn} becomes $\sum_i\lambda_ib_i\phi_i=\sum_ia_i\phi_i$, and by orthogonality of the $\phi_i$'s we have $b_i=\nicefrac{a_i}{\lambda_i}$; of course if $\lambda_i=0$ then $\phi_i$ satisfies the Laplace equation and can be added with a free coefficient.  In this way, solving the Poisson equation is trivial in the Laplace eigenfunction basis.

\subsection{The Parabolic Case:  The Heat Equation}

By introducing a time variable $t$, we can model PDEs that evolve from some starting condition. For instance, consider the time evolution of $u$ given by
\begin{equation}
u_t = -\Delta u
\end{equation}
where we denote $u_t=\frac{\partial u}{\partial t}$, where $u$ satisfies $u|_{t=0}\equiv f\in C^\infty(\Omega)$ and we constrain $u|_{\partial\Omega}=0\ \forall t\geq0$.

This PDE is known as the \emph{heat equation}, because it models the diffusion of $f$ over $\Omega$ for time $t$.  It is a model \emph{parabolic} equation, because the matrix of high-order derivatives is positive semi-definite (only first derivatives appear in $t$).  Similar to the Poisson case, if we write $f=\sum_ia_i\phi_i$ then it is easy to see that $u=\sum a_ie^{-\lambda_i t}\phi_i$ satisfies our conditions.

Despite the generality of our choice of $\Omega$, we actually can prove some facts about solutions of the heat equation:
\begin{lemma}[Weak Maximum/Minimum Principle]\label{weakmaximumprinciple}
Suppose $u\in C^\infty(\Omega)$ satisfies $u_t\geq -\Delta u$ with $u|_{\partial\Omega}\geq0\ \forall t\in[0,T)$ and $u|_{t=0}\geq 0$.  Then, $u\geq0\ \forall t\in[0,T]$.
\end{lemma}
\begin{proof}
Our proof follows Theorem 3.1 of~\cite{chung07} but is a standard proof in analysis.  Take $T_0\in(0,T)$ and $\varepsilon>0$, and define $v=u+\varepsilon t\in C^\infty(\Omega\times[0,T))$.  Suppose $v$ has a minimum at $(x_0,t_0)\in\Omega\times[0,T_0]$ with $x_0\in\mathring{\Omega}$ and $t_0\in(0,T_0]$.  Then, we must have
\begin{equation}\label{leqeqn}
v_t(x_0,t_0) = u_t(x_0,t_0)+\varepsilon \leq 0
\end{equation}
because $(x_0,t_0)$ is a stationary point of $v$; we only require $v_t\leq0$ rather than $v_t=0$ in case $t_0=T_0$.  By the ``categorizes extrema'' property, we also have
\begin{equation}\label{geqeqn}
\Delta v(x_0,t_0) = \Delta u(x_0,t_0) \leq0
\end{equation}
Combining~\eqref{leqeqn} and~\eqref{geqeqn} yields
\begin{equation}
v_t(x_0,t_0) +\Delta v(x_0,t_0)\leq 0.
\end{equation}
This contradicts the relationship
\begin{equation}
v_t(x_0,t_0) +\Delta v(x_0,t_0) = [u_t(x_0,t_0) +\Delta u(x_0,t_0)] + \varepsilon \geq \varepsilon > 0
\end{equation}
by the premises of the Lemma.  Thus, we must have that the only minima of $v$ occur when $t_0=0$ or $x_0\in\partial\Omega$.  Suppose the minimum occurs at $(x_0,t_0)$.  For generic $(x,t)$, we have 
\begin{align*}
u(x,t)&=v(x,t)-\varepsilon t\textrm{ by definition of }v\\
&\geq v(x_0,t_0)-\varepsilon t\textrm{ by definition of }(x_0,t_0)\\
&=u(x_0,t_0)+\varepsilon t_0 -\varepsilon t\textrm{ again by definition of }v\\
&\geq\varepsilon(t_0-t)\textrm{ because }u\geq 0\textrm{ on }\partial\Omega\textrm{ and when }t=0\\
&\geq-\varepsilon T\textrm{ since }t_0,t\in[0,T]
\end{align*}
Note that $T$ is fixed and $\varepsilon>0$ is arbitrary, so in reality this inequality shows $u(x,t)\geq0$, as needed.
\end{proof}

This principle leads us to the uniqueness of solutions to the heat equation.  In particular, if $u_1$ and $u_2$ both satisfy the heat equation with the same boundary conditions, then $u_{\textrm{diff}}\equiv u_1-u_2$ satisfies the heat equation with zero initial and boundary conditions.  By the maximum principle applied to $u_{\textrm{diff}}$ and $-u_{\textrm{diff}}$, we must have $u_{\textrm{diff}}=0$ for all $t$, as needed.  In fact, letting $t\rightarrow\infty$ approaches solutions to the Laplace equation $\Delta u=0$, which thus has the unique solution $u\equiv0$ for our zero boundary conditions.  This observation in turn leads to uniqueness of Poisson equation solutions $\Delta u = g$ since two solutions $u_1$ and $u_2$ satisfy $\Delta(u_1-u_2)=0$.

\subsection{The Hyperbolic Case:  The Wave Equation}

Our final model second-order PDE is the \emph{wave equation}, given by
\begin{equation}
u_{tt}=-\Delta u
\end{equation}
with the same boundary conditions as the heat equation.  We also prescribe $u_t=g$ at $t=0$.  Once again we can compute solutions in closed form as $u=\sum_i a_i\cos(\sqrt{\lambda_i}t+b_i)\phi_i$, where $a_i,b_i$ are chosen to satisfy the given initial conditions; if $\lambda_i=0$ then we replace $\cos(\cdots)$ with $a_i+b_it$.  It is \emph{hyperbolic} because the time derivative has a different sign than the space derivative.  We defer discussion of the uniqueness of wave equation solutions until we have developed appropriate notions of the gradient of a function, which depends more strongly on the domain $\Omega$ in question.

\begin{figure}
\centering
\begin{tabular}{|l|l|p{3.4in}|}\hline
\textbf{Symbol} & \textbf{Description} & \textbf{Notes}\\\hline
$G=(V,E)$ & Graph & $V$ is the set of vertices; $E\subseteq V\times V$ is the set of edges\\
$l_e$ & Edge length & Associated with edge $e\in E$\\
$\G$ & Geometric realization & Constructed by gluing intervals $\subset\R$ of length $l_e$ using the topology of $G$\\
$C^k(\G)$ & Differentiable functions & $C^k$ on edge interiors; $C^0$ at vertices \\
$T\G$ & Tangent bundle & Union of tangent bundles of edge interiors\\
$X$ & Vector field & Function $\G\backslash V\rightarrow T\G$\\
$\nabla f$ & Gradient & Vector field on $\G\backslash V$ constructed using 1D calculus\\
$\nabla_\calc \cdot X$ & Calculus divergence & Scalar function on $\G\backslash V$\\
$\V$ & Vertex measure & Discrete measure on $\G$ with $\V(v)=1\ \forall v\in V$\\
$\E$ & Edge measure & Lebesgue edge measure on $\G$; $\E(e)=l_e\ \forall e\in E$\\
$d\Gamma$ & Integrating factor & Written $\alpha\ d\V + \beta\ d\E$ for integration using $\V$ and $\E$\\
$d\D_X$ & Divergence & Integrating factor with $\nabla_\calc$ on $E$ and discrete terms on $V$\\
$d\L_f$ & Laplacian & Integrating factor given by $-d\D_{\nabla f}$\\
$\Delta_E$ & Edge-based Laplacian & Coefficient of $d\E$ in $d\L_f$ given by $-\nabla_\calc\cdot\nabla f$\\
$\Delta_V$ & Vertex-based Laplacian & Coefficient of $d\V$ in $d\L_f$ given by $\tilde{n}\cdot\nabla f$\\
$\tilde{A}$ & Normalized adjacency matrix & The adjacency matrix of $G$ with rows normalized to sum to $1$ \\
$A^h$ & Set expansion & For $A\subseteq\G$, $A^h=\{x\in\G : \textrm{dist}(x,A)<h\}$\\
\hline
\end{tabular}
\caption{Notation for our discussion of~\cite{friedman04}.}\label{friedmannotation}
\end{figure}

\section{Geometric Realizations of Graphs}

We turn now to the case when $\Omega$ is a graph $G=(V,E)$ for vertices $V$ and undirected edges $E$.  We assume $|V|,|E|<\infty$ and that each edge $e\in E$ is associated with a length $l_e>0$.  This description, however, is a theoretical construct, so we must describe what it means to diffuse heat or propagate waves along a graph; different models of diffusion or differential structures will lead to different adaptations of the model PDEs from \S\ref{basicanalysis}.

We commence our discussion with the most intuitively-appealing construction, which is able to apply analytical results directly but in the end suffers from considerable mathematical challenges.  Our development will follow that of~\cite{friedman04}; we summarize notation in Figure~\ref{friedmannotation}.

\subsection{Preliminaries}

By far the most common visualization of a graph assigns a point in $\R^n$ to each vertex in $V$ and connects the pairs in $E$ using line segments or curves of length $l_e$.  Note that such an embedding is a \emph{construction} rather than a fundamental property of graphs.  For instance, there is no straightforward reason why a social network graph should be embeddable in a low-dimensional space in a meaningful way.  Even so, much of our intuition for graph theory comes from such a construction.

Formalizing this notion, we define the \emph{geometric realization} of $G$ as a topological object $\G$ consisting of a closed interval of length $l_e$ for each edge $e\in E$ whose endpoints are identified with those of other segments using the connectivity of $G$.  We identify each $e\in E$ with its corresponding segment within $\G$ and each $v\in V$ with the corresponding point; we will use ``vertex'' and ``edge'' to refer to either type interchangeably when the difference is unambiguous.

The maximal open interval contained within each edge $e$ of $\G$ is the \emph{interior} of $e$.  Note that edge interiors are simply line segments and as such can be modeled to conduct heat or waves identically to segments of $\R$.  Thus, the primary challenge of this setup is the treatment of differential operators near the points in $V$, which are nonmanifold when vertices have degree $\neq 2$.  Thus, to take advantage of this structure, before developing PDE theories for $\G$, we develop its global differential structure.

We first define coordinate-free differentiation of scalar functions on $\mathcal{G}$.  Take $T\G$ to be the ``tangent bundle'' of $\G$ given by the union of the tangent bundles of the edge interiors, isomorphic to $(\G\backslash V)\times\R$.  Then, for functions $f\in C^k(\G)=C^k(\G\backslash V)\cap C^0(\G)$ and vector fields $X\in C^k(\G\backslash V,T\G)$ we can define the \emph{gradient} $\nabla f\in C^{k-1}(\G\backslash V,T\G)$ and \emph{calculus divergence} $\nabla_\calc\cdot X\in C^{k-1}(\G\backslash V)$ by applying the usual definitions of gradient and divergence within edge interiors.

It is tempting to define the Laplacian of a function on $\G$ in a similar fashion to complete our arsenal of differential graph operators.  Unfortunately, doing so trivializes the heat and wave equations, as we have not specified any way for information to be communicated across the vertices.  Thus, our development must be considerably more subtle.

\subsection{Integration and Divergence}

Our adaptation of differential operators thus far has considered the one-dimensional structures--or edges--of $\G$ but has neglected the zero-dimensional vertices.  Interestingly, vertices of $\G$ appear in our special treatment of the non-manifold points in $\G$ in a form similar to boundary conditions of PDEs on intervals in $\R$.  This observation makes sense in light of the fact that we can consider diffusion and wave propagation along $\G$ to satisfy the usual equations within edge interiors, with boundary conditions that their values match at vertices.

To facilitate subsequent discussion, we introduce two measures on $\G$:
\begin{itemize}
\item A discrete \emph{vertex measure} $\V$ satisfying $\V(v)=1\ \forall v\in V$.
\item A \emph{edge measure} $\E$ coinciding with the usual Lebesgue measure on edge interiors.
\end{itemize}
Note that this definition is somewhat more restrictive than that in~\cite{friedman04}, but our results generalize easily to their more flexible setting.  We denote integration against these two measures using $\int\cdots d\V$ and $\int\cdots d\E$, resp.

Relating integrals against $\V$ and $\E$ will provide a compact methodology for expressing our primary results about PDEs on $\G$.  Thus, we introduce \emph{integrating factors} $d\Gamma=\alpha\ d\V + \beta\ d\E$, for $\alpha:\V\rightarrow\R$ and $\beta\in L_2(\G\backslash V)$; we may introduce continuity or differentiability conditions on $\beta$ as need be.  Integration of functions on $\G$ is carried out as:
\begin{equation}
\int_\G f\ d\Gamma = \int f\alpha\ d\V + \int f\beta\ d\E
\end{equation}

As an initial example of the types of operations we can handle using this setup, we develop a notion of vector field divergence on $\G$ satisfying a theorem analogous to the divergence theorem on $\R^n$.  In general, note that our ``calculus divergence'' is somewhat unsatisfactory in that it yields values only on $\G\backslash V$.  While this restriction might be acceptable for directional vector fields, which most naturally are associated with points on edges, divergence may be measurable at vertices as well as on edges.

We have two potential strategies for relating a given test function $g\in C^\infty(\G)$ and the divergence of a vector field $X$: integrating the product $g\nabla\cdot X$ or integrating the product $-X\cdot\nabla g$.  Applying intuition from integration by parts, these two should be equivalent for the proper choice of $\nabla$.  Unfortunately, carrying out the first expression with the calculus divergence does \emph{not} yield a divergence theorem for graphs, so we explore the second.  Define a functional $\D_X:C^\infty(\G)\rightarrow\R$ given by
\begin{equation}
\D_X(g) = -\int_{\G\backslash V} X\cdot \nabla g\ d\E
\end{equation}
We prove a useful formula for such an operator:
\begin{lemma}[\cite{friedman04}, Proposition 2.16]\label{divergencelemma} Suppose we define $\hat{e}$ to be the unit vector from $u$ to $v$ paralleling edge $e=(u,v)\in E$ in $T\G|_e$ and take
\begin{equation}
(\tilde{n}\cdot X)(u) = \sum_{e=(v,u)\in E} \hat{e}\cdot X|_e(v).
\end{equation}
Then, for $X\in C^1(T\G)$ and $g\in C^\infty(\G)$,
\begin{equation}
\D_X(g) = \int \left[(\nabla_\calc\cdot X)g\ d\E - (\tilde{n}\cdot X)g\ d\V\right],
\end{equation}
where we assume $X$ can be extended to $\partial e=\{u,v\}$ by continuity.
\end{lemma}
\begin{proof}
By Stokes' Theorem, if we restrict to a single edge $e=(u,v)$ we clearly have the relation
\begin{displaymath}
\int_e \nabla_\calc \cdot X\ d\E = \hat{e}\cdot (X|_e(v)-X|_e(u)),
\end{displaymath}
where the second term represents a (discrete) integral about the boundary of $e$.  Summing over all edges in $\G$ yields
\begin{displaymath}
\int_\G \nabla_\calc\cdot X\ d\E = \int_\G \tilde{n}\cdot X\ d\V
\end{displaymath}
Consider the vector field $Xg$.  Applying this formula yields:
\begin{align*}
\int_\G \nabla_\calc\cdot(Xg)\ d\E&=\int_\G \tilde{n}\cdot (Xg)\ d\V\\
&= \int_\G (\tilde{n}\cdot X)g\ d\V\textrm{ by linearity of the dot product}
\end{align*}
By the chain rule we can write $\nabla_\calc \cdot (Xg) = (\nabla_\calc \cdot X)g + X\cdot \nabla g$.  Substituting this expression into the integral above and reordering terms yields the desired result.
\end{proof}

Thus, if we define the divergence of $\G$ to be the integrating factor
\begin{equation}\label{divergence}
d\D_X=(\nabla_\calc\cdot X)d\E - (\tilde{n}\cdot X)d\V,
\end{equation}
then we satisfy the formula
\begin{equation}
\int_\G g\ d\D_X + \int_\G X\cdot\nabla g\ d\E = 0.
\end{equation}
This formula serves as our divergence theorem for geometric realizations of graphs.  In particular, if we take $g\equiv1$ then we find
\begin{equation}
\int_\G \ d\D_X=0,
\end{equation}
as expected since the graphs we consider satisfy $\partial\G=\emptyset$.~\cite{friedman04} shows how to generalize this result when a subset of $V$ is marked as $\partial\G$.

\subsection{Geometric Graph Laplacian}\label{geometricgraphlaplacian}

We now can define the Laplacian of a graph's geometric realization:
\begin{equation}
d\L_f = -d\D_{\nabla f}
\end{equation}
Note that in classical notation this definition coincides with the relation $\Delta f = -\nabla\cdot(\nabla f)$, and that our choice of Laplacians parallels that of~\cite{friedman04}.  We define ``edge-based'' and ``vertex-based'' Laplacians $\Delta_E$ and $\Delta_V$ by requiring $d\L_f=\Delta_Ed\E+\Delta_Vd\V$; substituting~\eqref{divergence} yields:
\begin{align}
\Delta_E f &= -\nabla_\calc\cdot\nabla f\\
\Delta_V f &= \tilde{n}\cdot \nabla f
\end{align}
Note that a straightforward consequence of Lemma~\ref{divergencelemma} is the relationship
\begin{equation}\label{laplacianchangeofvariables}
\int g\ d\L_f = \int \nabla f\cdot \nabla g\ d\E = \int f\ d\L_g
\end{equation}
for sufficiently differentiable $f,g$.  It is also easy to see that $\Delta_Ef=0$ implies that $f$ must be edgewise-linear; in this case the vertex-based Laplacian $\Delta_V$ coincides with proposed discrete Laplacian matrices for graphs.  

Before modeling with $d\L_f$, we wish to understand the structure of its eigenspaces.  Of course, $d\L_f$ encodes two separate operators $\Delta_E$ and $\Delta_V$, whose eigenstructures we must understand separately.  Considering $\Delta_V$ or $\Delta_E$ alone is insufficient, since the former operates only at vertices and the latter on edge interiors; instead, we will find that appropriate eigenvalue problems \emph{couple} the two operators, using $\Delta_V$ to put boundary conditions on $\Delta_E$.

Note that $\G\backslash V$ is a collection of open intervals, which clearly does not admit a countable and complete set of eigenpairs $(f_i,\lambda_i)$.  Thus, we must add more conditions to make our problem well-posed.  We provide an alternative motivation based on a variational problem.

Suppose we wish to find critical points of the \emph{Rayleigh quotient}
\begin{displaymath}
\mathcal{R}(f) = \frac{\int |\nabla f|^2\ d\E}{\int |f|^2\ d\E}
\end{displaymath}
It is easy to see that uniformly scaling $f$ has no effect on $\mathcal{R}$, so we equivalently find critical points of $\bar{\mathcal{R}}(f)=\int|\nabla f|^2\ d\E$ subject to $\int|f|^2\ d\E=1$.  Such a variational problem leads to the Lagrange multiplier function:
\begin{displaymath}
\Lambda(f,\lambda)=\bar{\mathcal{R}}(f)-\lambda\left(\int|f|^2\ d\E-1\right)
=\int(|\nabla f|^2-\lambda|f|^2)\ d\E + \lambda
\end{displaymath}
Taking the G\^ateaux derivative with respect to $f$ in the $g$ direction yields:
\begin{align*}
d\Lambda(f,\lambda;g)
&= \frac{d}{dh}\Lambda(f+hg,\lambda)|_{h=0}\\
&= \int \frac{d}{dh}(|\nabla f+h\nabla g|^2-\lambda|f+gh|^2)|_{h=0}\ d\E\\
&= \int (2\nabla g\cdot\nabla f - 2\lambda fg)\ d\E\\
&= 2\int g\ d\L_f - 2\lambda \int fg\ d\E\textrm{ by~\eqref{laplacianchangeofvariables}}\\
&= 2\int g (\Delta_E f - \lambda f)\ d\E + 2\int g\Delta_V f\ d\V\textrm{ by definition of }d\L_f 
\end{align*}
So, if $f$ with $\int|f|^2\ d\E$ is a critical point of $\bar{\mathcal{R}}$, then it must satisfy both of:
\begin{align}
\Delta_E f &= \lambda f\\
\Delta_V f &= 0
\end{align}
for some $\lambda\in\R$.

This derivation motivates the definition of an \emph{edge-based eigenpair} of $d\L_f$ as a pair $(f,\lambda)$ satisfying $f\in C^\infty(\G)$, $\Delta_E f=\lambda f$, and $\Delta_V f = 0$.  Restricting $f$ to an edge interior of $\G$ gives $f(x)=A\cos(\omega x+B)$ for parameter $x\in[0,l_e]$ and $\omega=\sqrt{\lambda}$, since $f$ is an eigenfunction of the one-dimensional Laplacian along edges.  More globally, our Rayleigh quotient construction allows for a proof similar to that in~\cite{evans10} for PDEs on Euclidean domains of the following spectral theorem:

\begin{lemma}[\cite{friedman04}, Proposition 3.2]\label{spectraltheorem} The edge-based Laplacian admits a sequence of eigenpairs $(f_i,\lambda_i)$ satisfying
\begin{itemize}
\item $0\leq \lambda_1\leq\lambda_2\leq\cdots$
\item $\{f_i\}_{i\in\Z^+}$ forms a complete orthonormal basis for $L^2(\G)$, where orthonormality is measured with respect to $d\E$
\item $\lambda_i\rightarrow\infty$
\end{itemize}
\end{lemma}
In fact, it is straightforward to extend Weyl's Law to show that the number $N$ of eigenvalues $\leq\lambda$ for fixed $\lambda$ grows like $\sqrt{\lambda}$.

Returning to properties specific to the geometric graph Laplacian, we find that the $\Delta_V f=0$ condition provides a strong characterization of edge-based eigenfunctions.  In particular, for eigenvalue $\lambda>0$, recall that the corresponding eigenfunction $f$ restricted to an edge $e=(u,v)$ is given by $f|_e(x)=A\cos(\omega x+B)$ for $x\in[0,l_e]$, $\omega=\sqrt{\lambda}$, and some $A,B\in\R$.  We evaluate $f|_e$ and its derivatives at the endpoints:
\begin{align*}
f(u)&=f_e(0)=A\cos B\\
f(v)&=f_e(l_e)=A\cos(\omega l_e+B)=A\cos(\omega l_e)\cos(B)-A\sin(\omega l_e)\sin(B)\\
f'(0)&=-A\omega\sin B=-\omega\frac{f(v)-\cos(\omega l_e)f(u)}{\sin(\omega l_e)}\textrm{ by combining the last two expressions}
\end{align*}
Summing the final expression around $u$ is exactly the vertex-based Laplacian at $u$, yielding the following condition for values of $f$ at vertex $u$:
\begin{equation}\label{zerovertexcondition}
0=\Delta_V f|_u=\sum_{e=(u,v)\in E} \frac{f(v)-\cos(\omega l_e)f(u)}{\sin(\omega l_e)}
\end{equation}
In fact, any $\omega$ satisfying~\eqref{zerovertexcondition} is an edge-based eigenvalue, since we can use our cosine expression for $f|_e$ to fill in the remaining eigenfunction in edge interiors.  In other words, we have proved the following lemma:
\begin{lemma}
$(f,\omega^2)$ with $\omega l_e\neq k\pi$ for any $e\in E,k\in\Z$ is an edge-based eigenpair if and only if $\omega$ satisfies~\eqref{zerovertexcondition} for all $u\in V$ and $\Delta_E f = \omega^2 f$.
\end{lemma}

This expression is encouraging computationally in that it gives a finite system of equations for finding edge-based eigenvalues.  Unfortunately, the condition is highly nonlinear and nonconvex in $\omega$, so we cannot in general find the complete spectrum of a graph with arbitrary edge lengths.

One case in which it is possible to find a closed form for the edge-based spectrum of $\G$ is when $l_e=1\ \forall e\in E$.  Then, the denominator of~\eqref{zerovertexcondition} has no dependence on $e$ and can be removed, leaving behind a much simpler expression that can be solved for $\omega$.  

\begin{lemma}[\cite{friedman04}, Proposition 3.5]\label{adjacencylemma}
Let $\tilde{A}$ be the adjacency matrix of $G$ after normalizing each row to sum to $1$.  Then, for each eigenvalue $\lambda$ of $\tilde{A}$, $\arccos(\lambda)+2\pi k$ and $2\pi-\arccos(\lambda)+2\pi k$ are edge-based eigenvalues for $k\in\Z^+$ with the same multiplicities.  Furthermore, eigenvalues $\pi+\pi k$ occur with multiplicity $|E|-|V|$; if this value is negative, we \emph{subtract} the multiplicity from those occurring in $\tilde{A}$.  
\end{lemma}

We omit the analytical details of the proof.

\subsection{Geometric Graph Wave Equation}

Having developed differential operators on $\G$ up to second order, we are now equipped to introduce the general wave equation on $\G$:
\begin{equation}\label{geometricwave}
(\alpha d\V + \beta d\E)u_{tt} = d\D_{\gamma\nabla u}
\end{equation}
Here, we make the following assumptions:
\begin{itemize}
\item $u\in C^0(\G\times I)$ with $u(\cdot,t)\in C^2(\G)$ for some interval $I\subseteq \R$
\item $\alpha,\beta\in C^0(\G)$ and $\gamma\in C^1(\G)$ with $\alpha,\beta,\gamma\geq0$
\end{itemize}
Note if $\gamma\equiv1$ then the equation becomes $(\alpha d\V+\beta d\E)u_{tt}=-d\L_u$, mimicking the classical wave equation.  Splitting into $d\V$ and into $d\E$ terms yields the following system of equations:
\begin{align}
\alpha u_tt &= -\gamma\Delta_V u\\
\beta u_tt &= -\gamma\Delta_E u + \nabla\gamma\cdot\nabla u
\end{align}

To prove basic properties of~\eqref{geometricwave}, we will examine the energy function given by
\begin{equation}\label{geometricwaveenergy}
\textrm{Energy}(A;t)=\int_A (\gamma(\nabla u)^2\ d\E + u_t^2(\alpha\ d\V+\beta\ d\E))
\end{equation}
For $A\subseteq\G$ define $A^h=\{x\in\G : \textrm{dist}(x,A)<h\}$.  The following theorem characterizes the evolution of this energy function over time:

\begin{lemma}
Suppose $\gamma\leq c^2\beta$ on $\G$.  Then, if $u$ satisfies~\eqref{geometricwave} then $\textrm{Energy}(A^{ct_0};0)\geq\textrm{Energy}(A;t_0).$
\end{lemma}

The proof of this theorem once again follows that of analogous results from classical PDEs.  It immediately implies uniqueness of solutions to the wave equation given $u$ and $u_t$ at $t=0$.  It also shows that if two solutions agree on some subset of $\G$, then they are guaranteed to agree on a subset of $\G$ shrinking at rate $c$ over time.  That is, the wave equation~\eqref{geometricwave} with $\beta>0$ can be associated with a speed of propagation $c$ at which information can be communicated along $\G$; this property does not hold for most formulations of the wave equation on graphs.  For example, taking $\beta=0$ corresponds to an edgewise-linear model of wave propagation, and it can be shown that a ``hat function'' about any vertex in $\G$ can propagate with infinite speed in this case.

Using arguments similar to our earlier discussion, it is easy to see that when all edges of $\G$ have the same length we can find a closed-form solution to the wave equation without considering edge interiors.  Such a solution technique uses the Chebyshev polynomials in the matrix $\tilde{A}$ to expand $\cos(n\tilde{A})$ for $n\in\Z^+$.

\subsection{Discussion and Applications}

\cite{friedman04a,friedman04} provide many theoretical applications of the geometric graph Laplacian and other differential structures on $\G$.  Spectral bounds on the diameter of $G$ and related quantities can be derived by applying analytical results involving Laplacian eigenvalues.  Additional results on graph structure can be derived by studying the heat equation with respect to the vertex Laplacian, by considering the class of edgewise-linear functions~\cite{friedman04a}; a more generic heat equation similar to~\eqref{geometricwave} is not considered.

This construction of ``calculus on graphs'' is very similar to the idea of a ``quantum'' or ``metric'' graph appearing in the theoretical physics literature paired with Neumann or natural matching boundary conditions~\cite{gnutzmann06}.  While some results from either field may be applicable to the other, the physical interpretations of  PDEs on graphs are different in the two cases.  For instance, quantum graphs unsurprisingly involve the wave equations appearing in quantum physics, whereas our constructions may be better interpreted as more classical oscillations of a graph realized using a series of wires.

Such geometrizations appear to be useful tools for proving theoretical bounds on graph structures. In particular, many theorems regarding large-scale behavior of assorted differential equations carry over directly to the graph case with little to no modification, as we have seen studying Laplacian eigenfunctions and the propagation of waves.  As a computational tool, however, these constructions can be cumbersome.  In particular, when edge lengths are nonuniform, we have no straightforward characterization of Laplacian eigenfunctions or solutions of the wave equation and are forced to resort to approximation by chains of unit edges or finite elements simulation.

\section{Discrete Operators on Graphs}\label{discreteoperators}

Rather than posing PDE problems on the continuous topology given by $\G$, we can attempt to formulate \emph{analogous} PDE-type equations on the original discrete graph $G$.  In this case, rather than applying PDE theory directly, we use it to inspire potential fully-discrete results.  These techniques have the advantage that they generally involve finite-dimensional linear algebra in $\R^{|V|}$ rather than more complex differential theory, even if Laplacians cannot be derived as easily from continuous theory.

\subsection{Graph Laplacians}\label{graphlaplacians}

Taking a reverse perspective for deriving potential graph Laplacians, we can examine patterns that appear in discretizations of continuous PDEs, in hopes that the Laplacian matrices from these domains have similar properties to the continuous operators that appear in the limit.  For instance, suppose we approximate a one-dimensional $u:\R\rightarrow\R$ with a sequence of values $\ldots,u_{-1},u_0,u_1,\ldots$.  Simple finite differencing approximates $\nicefrac{du}{dx}$ with the sequence differences $u_i-u_{i-1}$.  Taking differences of these differences yields a second derivative $u_{i+1}-2u_i+u_{i-1}$, or convolution with a filter given by $(\begin{tabular}{ccc}1 & -2 & 1\end{tabular})$.  A similar approach on an $n$-dimensional grid yields a weight of $-n$ on each central vertex and $1$ on each neighbor.  This construction suggests a ``combinatorial'' Laplacian for per-vertex functions given as vectors in $\R^{|V|}$, given by
\begin{equation}
L(u,v)=
\left\{
\begin{array}{ll}
\degree u & \textrm{ if }u=v\\
-1 & \textrm{ if $u$ and $v$ are adjacent}\\
0 & \textrm{ otherwise}
\end{array}
\right.
\end{equation}

Similarly, consider random walks along a graph $G$ with transition matrix $P$.  That is $P(u,v)$ is the probability that state $u$ transitions to state $v$ where $u$ and $v$ are adjacent in the graph, and otherwise $P(u,v)=0$.  A positive vector $x\in\R^{|V|}$ summing to $1$ can represent the probability that a particle is at one of the vertices of $G$; $Px$ then gives the distribution of potential particle positions after one time step.

One might study \emph{stationary distributions} of such a Markov process, that is, probability distributions on the vertices of a graph that do not change after a step of the random walk.  These distributions are given by the equation $Px=x$, or equivalently $(I-P)x=0$.  Notice the form of the matrix $I-P$:
\begin{equation}
[I-P](u,v)=
\left\{
\begin{array}{ll}
1 & \textrm{ if }u=v\\
-P(u,v) & \textrm{ if $u$ and $v$ are adjacent}\\
0 & \textrm{ otherwise}
\end{array}
\right.
\end{equation}
In particular, if a vertex $u$ has equal probability of transitioning to any of its neighbors, this matrix equal to $T^{-1}L$, where $T$ is the diagonal matrix of vertex degrees.

Motivated by these two examples and using the notation of~\cite{chung00}, for a graph $G=(V,E)$ with edge weights $w_e$ we define the \emph{combinatorial Laplacian} matrix $L$ of $G$ to be:
\begin{equation}
L(u,v)=
\left\{
\begin{array}{ll}
d_u & \textrm{ if }u=v\\
-w_{(u,v)} & \textrm{ if $u$ and $v$ are adjacent}\\
0 & \textrm{ otherwise}
\end{array}
\right.
\end{equation}
with $d_u=\sum_{(u,v)\in E} w_{(u,v)}$.  Using our above example and analogous appearances of the Laplacian in continuous probability theory, we also define a \emph{discrete Laplacian} $\Delta=T^{-1}L$, for a diagonal matrix $T$ containing the values $d_u$.  Neither $L$ nor $\Delta$ is necessarily symmetric, so we define a symmetric \emph{normalized Laplacian} matrix $\LL$ satisfying $\LL=T^{1/2}\Delta T^{-1/2}=T^{-1/2} LT^{-1/2}$.

We take $S\subseteq V$ to be a subset of vertices of the graph and say a function $g\in\R^{|V|}$ satisfies \emph{Dirichlet} boundary conditions if it is zero on $V\backslash S$.  We denote the Dirichlet eigenvalues $\lambda_1\leq\lambda_2\leq\cdots\leq\lambda_s$ of the graph, for $s=|S|$ to be eigenvalues of $\LL_S$, the submatrix of $\LL$ corresponding to the vertices in $S$.

\subsection{Green's Functions}\label{greensfunctions}

The Green's function of a PDE is its solution operator.  Specifically, suppose a PDE on $\R$ takes the form $Lu(x)=f(x)$ for unknown $u$ and linear differential operator $L$.  Then, the Green's function $G:\R^2\rightarrow\R$ (more precisely a distribution) provides solutions $u$ for any $f$ by convolution:
\begin{equation}
u(x)=\int G(x,s)f(s)\ ds
\end{equation}
Integration against $G$ effectively provides an ``inverse'' for $L$, although boundary conditions must be established to justify such a statement.

While rigorously establishing that $G$ provides a unique solution to a given PDE can be an analytical challenge, the construction and manipulations of \emph{discrete} Green's functions is much simpler.  In particular, operators like the graph Laplacians from \S\ref{graphlaplacians} are simply matrices in $\R^{|V|\times|V|}$, so their Green's functions intuitively can be computed using matrix inverses.

One way to understand the construction of discrete Green's functions is through eigen-analysis on the matrix $\LL_S$.  Imitating the proof for geometric Laplacians in \S\ref{geometricgraphlaplacian}, we can see that the minimal eigenvalue $\lambda_1$ of $\LL_S$ can be computed as:
\begin{equation}
\lambda_1 = \inf_f \frac{\sum_{x,y\in S\cup\delta S} (f(x)-f(y))^2w_{(x,y)}}{\sum_{x\in S} f^2(x)d_x}
\end{equation}
Here, we define the boundary $\delta S$ as those vertices in $V\backslash S$ adjacent to vertices in $S$.  If $S$ is connected and $\delta S$ is nonempty, then this expression shows $\lambda_1>0$, and thus we can construct an inverse $\GG$ of $\LL_S$.  We consider $\GG$ to be an operator on $\R^{|V|}$, extending to vertices in $V\backslash S$ using zeros.  Then, we can write a matrix $G=T^{-1/2}\GG T^{1/2}$ providing a discrete Green's function for $\Delta$; that is, it satisfies $G\Delta=I$ on $S$.

One can use Green's functions of elliptic operators like the Laplacian to interpret solutions to more complex PDEs.  For instance, suppose for $t\geq0$ we define the \emph{heat kernel} $\HH_t=e^{-t\LL_S}$ using matrix exponentiation.  Then, the following heat equation is satisfied:
\begin{equation}\label{graphheat}
\frac{d}{dt}\HH_t f=-\LL_Sf
\end{equation}
The Green's function in this case appears using the relationship $\GG=\int_0^\infty \HH_t\ dt$.  For this reason, discrete Green's functions are of interest as the basic building blocks for studying not only their associated operators but also other equations in which they appear.

While Green's functions may be best described using matrix inverses, algorithmically this description leads to inefficient techniques:  inverses of the sparse matrices $\LL_S$, $L_S$, and $\Delta_S$ are not guaranteed to be sparse, and we have provided no guarantees on their conditioning.  If we expand the matrix exponential in~\eqref{graphheat}, however, we see that large eigenvalues of $\LL$ are dampened quickly as time $t$ progresses, so it is possible to focus computations on eigenvalues of $\LL_S$ close to $0$.

Since $\LL_S$ is symmetric and positive definite, we can write its eigenfunctions as $\phi_1,\phi_2,\ldots,\phi_{|S|}$ with corresponding eigenvalues $0<\lambda_1\leq\lambda_2\leq\cdots\leq\lambda_{|S|}$.  Applying the usual eigenvector matrix decomposition element-by-element, we can write
\begin{equation}\label{greensonLLs}
\GG(x,y)=\sum_i \frac{1}{\lambda_i} \phi_i(x)\phi_i(y)
\end{equation}
Thus, as $\lambda\rightarrow\infty$, the influence of $\lambda$ on $\GG$ vanishes.  This heuristic argument aside, if we compute the whole spectrum of $\GG$ we can use it to construct exact solutions to vertex-based PDEs on graphs explicitly.

\subsection{Solving the Laplace and Poisson Equations on Graphs}

Suppose we wish to find $f$ satisfying $\Delta f=g$ for some $g$ with nontrivial values on $S\cup \delta S$.  We can use the eigenvalue constructions of \S\ref{greensfunctions} to find such a solution explicitly.  We first consider the Laplace equation $\Delta f = 0$ and then extend our solution to the more general $\Delta f=g$: 

\begin{lemma}[\cite{chung00}, Theorem 1]\label{greensfunctionform}
Suppose $\Delta f=0$ on $S$ with $f(x)=\sigma(x)$ for $x\in\delta S$.  Then, $f$ satisfies\footnote{As an aside, the exponent of $d_x$ is correctly $-\frac{1}{2}$ rather than $\frac{1}{2}$ in~\cite{chung00}.}
\begin{equation}
f(z)=d_z^{-1/2}
\sum_i\frac{1}{\lambda_i}\phi_i(z)
\sum_{\substack{x\in S\\(x,y)\in E\\y\in\delta S}}
d_x^{-1/2}\phi_i(x)\sigma(y)
\end{equation}
\end{lemma}
\begin{proof}
For convenience, as in~\cite{chung00} we denote $\tilde{f}=T^{1/2}f$ and note that $\Delta f=0$ if and only if $\LL_S\tilde{f}=0$.  Thus, it is straightforward to translate between solutions of Laplace's equation with the symmetrized Laplacian $\LL$ and those using the discrete Laplacian $\Delta$.

Since $\tilde{f}$ satisfies an equation in $\LL_S$, it makes sense to expand it in the eigenfunctions $\{\phi_i\}$ as $\tilde{f}=\sum_i a_i\phi_i$, where by orthogonality of the $\phi_i$'s we have $a_i=\langle\phi_i,\tilde{f}\rangle$.  For convenience, the boundary conditions can be encoded by a function $f_0(x)$ given by $\sigma(x)$ on $\delta S$ and $0$ otherwise; we will use subscripts of $S$ to denote a function restricted to $S$.

With the notation in place, it is possible to perform considerable simplification:
\begin{align*}
\lambda_ia_i
&= \lambda_i\langle\phi_i,\tilde{f}\rangle\textrm{ since }a_i=\langle\phi_i, f\rangle\\
&= \langle \LL_S\phi_i,\tilde{f}\rangle\textrm{ since }\LL_S\phi_i=\lambda_i\phi_i\\
&=\langle\LL_S\phi_i,(\tilde{f}-\tilde{f}_0)\rangle\textrm{ since }f_0=0\textrm{ on }S\\
&=\langle\phi_i,T^{-1/2}LT^{-1/2}(\tilde{f}-\tilde{f}_0)\rangle\textrm{ since }\LL=T^{-1/2}LT^{-1/2}\textrm{ and }\LL\textrm{ is symmetric}\\
&=\langle\phi_i,T^{-1/2}L(f-f_0)_S\rangle\textrm{ by definition of }\tilde{f}\\
&=\langle T^{1/2}\phi_i,\Delta(f-f_0)_S\rangle\textrm{ since }\Delta=T^{-1}L\textrm{ and }T\textrm{ is symmetric}\\
&=\langle T^{1/2}\phi_i,-(\Delta f_0)_S\rangle\textrm{ since }\Delta f=0\\
&=-\sum_{x\in S}\left(\sqrt{d_x}\phi_i(x)\cdot\frac{1}{d_x}\sum_{(x,y)\in E} (f_0(x)-f_0(y))\right)\textrm{ by expanding the inner product}\\
&=\sum_{x\in S}\sum_{\substack{(x,y)\in E\\y\in\delta S}}d_x^{-1/2}\phi_i(x)\sigma(y)\textrm{ since }f_0=0\textrm{ in }S\textrm{ and }f_0(x)=\sigma(x)\textrm{ on }\delta S
\end{align*}
We divide both sides by $\lambda_i$ and substitute into the eigenfunction expansion of $\tilde{f}$:
\begin{equation}
\tilde{f}=\sum_i a_i\phi_i=\sum_i\left(\frac{1}{\lambda_i}\sum_{\substack{(x,y)\in E\\y\in\delta S}}d_x^{-1/2}\phi_i(x)\sigma(y)\right)\phi_i
\end{equation}
Finally, we apply the relationship $\tilde{f}=T^{1/2}f$:
{\allowdisplaybreaks\begin{align}
f_S(z)&=\sum_i\left(\frac{1}{\lambda_i}\sum_{\substack{(x,y)\in E\\y\in\delta S}}d_x^{-1/2}\phi_i(x)\sigma(y)\right)d_z^{-1/2}\phi_i(z)\\
&=d_z^{-1/2}
\sum_i\frac{1}{\lambda_i}\phi_i(z)
\sum_{\substack{x\in S\\(x,y)\in E\\y\in\delta S}}
d_x^{-1/2}\phi_i(x)\sigma(y)
\end{align}}
This final refactorization completes the proof.
\end{proof}

\begin{lemma}[\cite{chung00}, Theorem 2]
Suppose $\Delta f=g$ on $S$ with $f|_{\delta S}=\sigma$, where $S$ is a connected subgraph.  Then, $f$ can be decomposed as a sum $f_1+f_2$, where $\Delta f_1=0$ with $f|_{\delta S}=\sigma$ and $f_2=Gg$ for Green's function $G$ of $\Delta$.
\end{lemma}
\begin{proof}
Take $f_1=f-f_2$.  Then, $f=f_1+f_2$ and $\Delta f_1=\Delta f - \Delta f_2=g - g = 0$, as desired.
\end{proof}

With this characterization of Poisson equation solutions, it is straightforward to write closed-form Green's functions and solutions to Laplace and Poisson equations on graphs with regular structures~\cite{chung00}.  Predictably, Green's functions of grid-structured graphs nearly coincide with the usual Discrete Fourier Transform basis.

\subsection{Green's Functions without Boundaries}

Our construction depended on the fact that the set $S$ of vertices on which we computed Green's functions satisfied $\delta S\neq\emptyset$.  When $\delta S=0$, $L=L_S$ no longer is full-rank, so we must deal with this case using some specialized constructions.

It is easy to see that the vector $\phi_0$ given by $\phi_0(k)=\sqrt{\nicefrac{d_k}{\sum_i d_i}}$ comprises the null space of $\LL$.  Since $\LL$ thus projects $\phi_0$ out of functions on the graph, we simply exclude it when constructing the Green's function $\GG$:
\begin{equation}
\GG\LL=\LL\GG=I-\phi_0\phi_0^\top
\end{equation}
Furthermore we constrain $\GG\phi_0\phi_0^\top=0$.  Since the matrix $\phi_0\phi_0^\top$ projects onto unit vector $\phi_0$, these constraints simply state that $\GG$ should project out the $\phi_0$ component and act as an inverse to $\LL$ on the subspace orthogonal to $\phi_0$.

With such a lack of boundary conditions, it is only possible to solve Poisson-type equations up to constant shift by $\phi_0$.  Parabolic equations such as the heat equation remain largely unchanged since they mostly depend on initial conditions.  Other problems including many involving Markov chains are solved exactly by projecting out $\phi_0$~\cite{aldous02}; construction of such problems is out of scope of the discussion at hand.

\subsection{Relationship to Geometric Graph Realization}

The development in this section relied fundamentally on the collection of discrete Laplacians defined in \S\ref{graphlaplacians}.  These Laplacians satisfy the properties stated in \S\ref{laplacianproperties}, and thus generic proofs that depend only on ``Laplacian-like'' structures apply equally well to these Laplacians as to the geometric Laplacians introduced in \S\ref{geometricgraphlaplacian}.

To align completely with continuous intuition, however, ideally one might expect discrete Green's functions and discrete Laplacian eigenfunctions to coincide \emph{exactly} with their counterparts on the geometric graph realization $\G$ restricted to graph vertices.  Such optimism likely would remain unsatisfied in the most general case:  it seems unlikely to be able to find solutions to continuous PDEs like the geometric graph wave equation knowing and making use only of solution values at discrete locations.  Even so, we can prove that in certain cases the geometric and discrete graph Laplacians have more structure in common than the most basic properties of Laplacian operators.

Recall Lemma~\ref{adjacencylemma}, which characterizes geometric graph eigenvalues in the unit edge length case.  These geometric eigenvalues were given by $\{\arccos(\lambda)+2\pi\Z^+,2\pi-\arccos(\lambda)+2\pi\Z^+:\lambda\in\Lambda(\tilde{A})\}$, where $\lambda$ represents an eigenvalue of the normalized adjacency matrix $\tilde{A}$.  Writing $\tilde{A}$ element-by-element yields
\begin{equation}
\tilde{A}(u,v)=
\left\{
\begin{array}{ll}
\frac{1}{\degree u} & \textrm{ if $u$ and $v$ are adjacent}\\
0 & \textrm{ otherwise}
\end{array}
\right.
\end{equation}
This matrix has exactly the same structure as the Markov matrix $P$ from \S\ref{graphlaplacians} when the probability of a transition from a vertex to any of its neighbors is uniform.  Eigenfunctions of $P$ and $I-P$ trivially are identical, and thus we can see that \textbf{in this case eigenfunctions of the discrete Laplacian $\Delta$ coincide with geometric Laplacian eigenfunctions restricted to $V$}.  Corresponding eigenvalues are related using the $\arccos\lambda$ nonlinearity, although it is important to remember that $\G$ has an infinite spectrum since there are multiple frequencies within a given edge yielding the same boundary values.

Thus, in the equi-length case the discrete Laplacian actually can be used to construct geometric PDE solutions.  As edge lengths become nonuniform, however, this relationship breaks down, mostly due to the nonconstant denominator in~\eqref{zerovertexcondition}.

\section{Semi-Discrete PDEs on Graphs}

In \S\ref{discreteoperators}, we defined discrete Laplacian operators for per-vertex functions on graphs.  These operators have key properties in common with Laplacians encountered in classical analysis and admit structures like Green's functions paralleling objects on $\R^n$ and on manifolds.  Given this relationship, it is a reasonable step to consider using \emph{discrete} Laplacians to model flows on graphs.  Intuitively, such models yield flows that are spatially discrete along the domain $\Omega$ but continuous in time $t$.

This construction might be considered a ``semi-discrete PDE,'' in the sense that the time variable remains continuous.  An alternative viewpoint is that these equations become $\R^{|V|}$-valued systems of \emph{ordinary} differential equations (ODEs) with a single independent variable $t$.  Both of these dual viewpoints are valuable:
\begin{itemize}
\item The PDE standpoint indicates that resulting solutions likely will have structures resembling heat flow, wave propagation, and other physical effects modeled using parabolic and hyperbolic equations.
\item The mathematics of ODEs provides straightforward characterizations of existence, uniqueness, and properties of solutions that are stronger than their PDE equivalents.  For example, the Picard-Lindel\"of Existence and Uniqueness Thorem for ODEs automatically guarantees that discrete graph heat and wave equations have solutions for all time $t\geq0$, effectively sidestepping the need for more specialized existence and uniqueness proofs needed for many PDEs~\cite{hirsch12}.
\end{itemize}

\cite{chung07} exhaustively lists solution techniques for first- and second-order semi-discrete PDEs on graphs including low-order driving terms and other special cases.  Here, we will parallel their development in the homogeneous case, characterizing solutions to the model equations of \S\ref{basicanalysis}.  Analogous techniques apply to the more general case, although notation and technicalities become considerably more complex.

\subsection{$\omega$-PDEs on Graphs}

We will continue to use the same notation as previous sections but briefly connect it to the notation used in~\cite{chung07}.  There, graph edges are associated with weights $\omega:V\times V\rightarrow[0,\infty)$; these are identical to the weights $w_{(x,y)}$ defined earlier, taking $w_{(x,y)}=0$ when $(x,y)\notin E$.  The prefix $\omega$ is added to PDE terminology to disambiguate from the continuous case.  Vertex degrees $d_\omega u$ are the same as our $d_u$, and integration of $f:V\rightarrow\R$ is given by $$\int_G f\ d_\omega\equiv \sum_{u\in V} f(u) d_{u}=\int fd_u\ d\V$$

In contrast to~\cite{friedman04,friedman04a}, \cite{chung07} develops a discrete calculus on graphs without constructing a geometric realization.  They define a directional derivative $[D_v f](u)$ of $f:V\rightarrow\R$ by writing differences
\begin{equation}
[D_v f](u)\equiv (f(v)-f(u))\sqrt{\frac{w_{(u,v)}}{d_u}}
\end{equation}
The gradient $\nabla f:V\rightarrow\R^{|V|}$ at $u$ is defined as a vector of $D_v f$ values evaluated at $u$:
\begin{equation}
(\nabla f)(u)\equiv (D_v f(u))_{v\in V}
\end{equation}
Second-order PDEs are modeled using the discrete Laplacian $\Delta$ introduced in \S\ref{graphlaplacians}.\footnote{The Laplacian used here differs from that in~\cite{chung07} by a sign.}  The definition of directional differentiation above is consistent with this Laplacian, in the sense that the Laplacian looks like summing double applications of $D_v$.  This self-consistent system admits analogs of many integration-by-parts and related vector calculus identities for discrete graph operators~\cite{chung05}.

\subsection{$\omega$-Diffusion Equations}

Although our definition of the Laplacian has changed, the form of the homogeneous $\omega$-diffusion or heat equation on discrete graphs remains the same: $$u_t=-\Delta u,$$ although now we have $u\in\R^{|V|}$ a signal only on vertices rather than $\G$.  Expanding the Laplacian row-by-row demonstrates that this equation models flows between nodes where conductivity is measured using $w$ and the instantaneous flow rate is determined by the difference in value of $u$ between adjacent nodes; such models appear in discretized PDEs as well as in models of electrical networks.  Notice that the proof of Lemma~\ref{weakmaximumprinciple} applies for this heat equation, implying not only existence and uniqueness but also a characterization of the continuum of values appearing in $u$.

Employing the notation of \S\ref{greensfunctions}, we will assume that a subset $S\subseteq V$ is specified on which we desire to solve the heat equation and that boundary values are prescribed on $\delta S$.  To do so, we once again consider the normalized Laplacian $\LL=T^{1/2}\Delta T^{-1/2}$ restricted to $S$, yielding operator $\LL_S$ with eigenvalues $0<\lambda_1\leq \lambda_2\leq\cdots\leq\lambda_{|S|}$ and eigenfunctions $\phi_1,\phi_2,\ldots,\phi_{|S|}$.  Since $\LL_S$ is positive definite, we can assume that the $\phi_i$'s are orthogonal and unit length.

We can use intuition from the construction of~\eqref{greensonLLs} and the fact that $e^{-\lambda_it}\phi_i$ satisfies the heat equation $\forall i$ to construct a heat kernel $K_S$:
\begin{equation}
K_S(u,v,t)=\sum_{i=1}^{|S|} e^{-\lambda_it}\phi_i(u)\phi_i(v)\sqrt{\frac{d_v}{d_u}}
\end{equation}
This kernel provides a succinct description of solutions to the graph heat equation:

\begin{lemma}[\cite{chung07}, Theorem 3.4]\label{discreteheatsolution}
Suppose $\delta S\neq\emptyset$.  Take functions $\sigma:\delta S\times[0,T)\rightarrow\R$ and $f:S\rightarrow\R$ such that $\sigma$ is continuous and $L^1$ in time $t$ for each vertex.  Then, the solution of the heat equation $u_t=-\Delta u$ with $u|_{t=0}=f$ and $u|_{\delta S}=\sigma\ \forall t\in[0,T)$ is given by
\begin{equation}
u(x,t) = \langle K_S(x,\cdot,t), f(\cdot)\rangle_S + \int_0^t \langle K_S(x,\cdot,t-\tau), B_\sigma(\cdot,\tau)\rangle\ d\tau
\end{equation}
where inner products are over graph vertices summed over ``$\cdot$'' and
\begin{equation}
B_\sigma(y,t)=\sum_{\substack{z\in\delta S\\(y,z)\in E}}\frac{\sigma(z,t)w_{(y,z)}}{d_y}
\end{equation}
\end{lemma}
\begin{proof}
Suppose $u$ is such a solution.  By completeness of the basis $\{\phi_i\}_i$, we can choose functions $a_i(t)=\langle T^{1/2}u,\phi_i\rangle$ satisfying
\begin{equation}\label{eigenexpansion}
T_S^{1/2}u|_S=\sum_i a_i(t)\phi_i(x)
\end{equation}
since each $\phi_i$ satisfies the zero boundary condition.

Since $\LL=T^{1/2}\Delta T^{-1/2}$, we can write $\LL T^{1/2}=T^{1/2}\Delta$.  Thus, in the style of the proof of Lemma~\ref{greensfunctionform}, we expand the inner product for $a_i$:
\begin{align*}
\lambda_ia_i(t)
&=\langle T^{1/2} u(\cdot,t), \lambda_i\phi_i\rangle_S\textrm{ by definition of }a_i\\
&=\langle T^{1/2} u(\cdot,t), \LL_S\phi_i\rangle_S\textrm{ since $\phi_i$ is an eigenfunction of $\LL_S$}\\
&=\langle T^{1/2} u(\cdot,t), \LL \phi_i\rangle_{\bar{S}} - \langle T^{1/2} u(\cdot,t),\LL\phi_i\rangle_{\delta S}\textrm{ where we define }\bar{S}=S\cup\delta S\\
&=\langle \LL T^{1/2} u(\cdot,t), \phi_i\rangle_{\bar{S}} - \langle T^{1/2} u(\cdot,t),\LL\phi_i\rangle_{\delta S}\textrm{ by symmetry of }\LL\\
&=\langle T^{1/2}\Delta u(\cdot,t), \phi_i\rangle_{\bar{S}} - \langle T^{1/2} u(\cdot,t),\LL\phi_i\rangle_{\delta S}\textrm{ since }\LL=T^{1/2}\Delta T^{-1/2}\\
&=-a_i'(t)-\langle T^{1/2} u(\cdot,t),\LL\phi_i\rangle_{\delta S}\textrm{ since $u$ satisfies the heat equation}\\
&=-a_i'(t)-\langle T^{1/2} \sigma(\cdot,t),\LL\phi_i\rangle_{\delta S}\textrm{ by boundary conditions on $u$}
\end{align*}
If we expand the definition of $\Delta$ row-by-row, it is easy to see that $$\Delta f(x)=\sum_{(x,y)\in E}(f(x)-f(y))\frac{w_{(x,y)}}{d_x}.$$   This implies the following form for applications of $\LL$:
\begin{align*}
\LL f(x) &= [T^{1/2}\Delta T^{-1/2}] f(x)\textrm{ by definition of }\LL\\
&=\sqrt{d_x}\sum_{(x,y)\in E} \left(\frac{f(x)}{\sqrt{d_x}}-\frac{f(y)}{\sqrt{d_y}}\right)\frac{w_{(x,y)}}{d_x}\textrm{ since we must apply powers of }T\\
&= \sum_{(x,y)\in E} \left(\frac{f(x)}{d_x} - \frac{f(y)}{\sqrt{d_xd_y}} \right)w_{(x,y)}
\end{align*}
In particular, taking $f=\phi_i$ and evaluating on the boundary $x\in\delta S$ yields
\begin{equation}
\LL \phi_i(x) = - \sum_{(x,y)\in E} \frac{\phi_i(y)w_{(x,y)}}{\sqrt{d_xd_y}}
\end{equation}
since $\phi_i$ satisfies the zero Dirichlet boundary conditions.

Thus, returning to our earlier chain of equalities for $\lambda_ia_i(t)$ shows
\begin{equation}
\lambda_ia_i(t) = -a_i'(t)+\sum_{\substack{z\in\delta S\\(y,z)\in E}}\frac{\sigma(z,t)\phi_i(y)w_{(y,z)}}{\sqrt{d_y}}
\end{equation}
This is an inhomogeneous first-order ordinary differential equation in $t$ whose solution is:
\begin{equation}\label{inhomo}
a_i(t) = c_i e^{-\lambda_i t} + e^{-\lambda_i t}\sum_{\substack{z\in\delta S\\(y,z)\in E}} \left(\int_0^t \sigma(z,\tau)e^{\lambda_i\tau}\ d\tau \right)\frac{\phi_i(y)w_{(y,z)}}{\sqrt{d_y}}
\end{equation}
for constants $\{c_i\}_{i=1,\ldots,|S|}\subset\R$.  Note $a_i(0)=c_i$ in this characterization.

To choose the constants $c_i$, we return to the definition of $a_i(0)$ and evaluate at $t=0$:
\begin{equation}\label{cexp}
c_i =\langle T^{1/2}u(\cdot,0),\phi_i\rangle\equiv \langle T^{1/2}f,\phi_i\rangle_S
\end{equation}

Finally, we substitute this expression into~\eqref{eigenexpansion} to show:
\begin{align*}
u(x,t)&= \frac{1}{\sqrt{d_x}}\sum_i a_i(t)\phi_i(x)\\
&=\frac{1}{\sqrt{d_x}}\sum_i \left[  c_i e^{-\lambda_i t} + e^{-\lambda_i t}\sum_{\substack{z\in\delta S\\(y,z)\in E}} \left(\int_0^t \sigma(z,\tau)e^{\lambda_i\tau}\ d\tau \right)\frac{\phi_i(y)w_{(y,z)}}{\sqrt{d_y}} \right]\phi_i(x)\textrm{ by~\eqref{inhomo}}\\
&=\frac{1}{\sqrt{d_x}}\sum_i \left[  \langle T^{1/2}f,\phi_i \rangle_S e^{-\lambda_i t} + e^{-\lambda_i t}\sum_{\substack{z\in\delta S\\(y,z)\in E}} \left(\int_0^t \sigma(z,\tau)e^{\lambda_i\tau}\ d\tau \right)\frac{\phi_i(y)w_{(y,z)}}{\sqrt{d_y}} \right]\phi_i(x)\textrm{ by~\eqref{cexp}}\\
&=\sum_{y\in S}K_S(x,y,t)f(y) + \int_0^t\left(\sum_{\substack{z\in\delta S\\(y,z)\in E}}K_S(x,y,t-\tau)\frac{\sigma(z,\tau)w_{(y,z)}}{d_y}\right)\ d\tau\textrm{ by definition of }K_S\\
&=\sum_{y\in S}K(x,y,t)f(y) + \int_0^t\langle K_S(x,\cdot,t-\tau), B_\sigma(\cdot,\tau)\rangle\ d\tau\textrm{, as desired.}
\end{align*}
This argument shows that if a solution exists to the heat equation then it must take this form.  It is easy to check that this expression does indeed satisfy the heat equation; in fact, this would comprise a short but somewhat obtuse proof of the lemma.  Nevertheless, it thus represents the unique solution to the heat equation with given boundary and initial conditions.
\end{proof}

The proof of Lemma~\ref{discreteheatsolution} is fairly technical, but from a high level it demonstrates the dual PDE/ODE approach to understanding flows on graphs.  ODE theorems are used to guarantee existence and simplify convergence since sums are finite.  The construction of the solution, however, relies on a separation of variables technique reminiscent of that used to solve the heat equation on subsets of $\R^n$.

\subsection{$\omega$-Elastic Equations}

As with the heat equation, the homogeneous $\omega$-elastic equation on discrete graphs takes the usual form $$u_tt=-\Delta u,$$ for signals $u\in\R^{|V|}$.  Note that this equation when viewed as an ODE models the behavior of a network of linked springs connecting vertices on a graph.

It is straightforward if verbose to parallel the proof of Lemma~\ref{discreteheatsolution} to construct closed-form solutions to the $\omega$-elastic equation with given boundary and initial conditions.  To do so, we will construct an additional Dirichlet kernel for this contrasting problem:
\begin{equation}
W_S(x,y,t) = \left[ t\phi_0(x)\phi_0(y) + \sum_i \frac{1}{\sqrt{\lambda_i}}\sin(\sqrt{\lambda_i}t)\phi_i(x)\phi_i(y)\right]\sqrt{\frac{d_y}{d_x}}
\end{equation}

Then, solutions are provided by the following lemma:
\begin{lemma}[\cite{chung07}, Theorem 4.3]\label{discretewavesolution}
Suppose $\delta S\neq\emptyset$.  Take functions $\sigma:\delta S\times[0,T)\rightarrow\R$ and $f,g:S\rightarrow\R$ such that $\sigma$ is continuous and $L^1$ in time $t$ for each vertex.  Then, the solution of the elastic equation $u_t=-\Delta u$ with $u|_{t=0}=f$, $u_t|_{t=0}=g$, and $u|_{\delta S}=\sigma\ \forall t\in[0,T)$ is given by
\begin{equation}
u(x,t) = \langle W_S(x,\cdot,t),g\rangle_S + \left\langle\frac{\partial}{\partial t}W_S(x,\cdot,t),f\right\rangle_S + \int_0^t \langle W_S(x,\cdot,t-\tau), B_\sigma(\cdot,\tau)\rangle_S\ d\tau
\end{equation}
\end{lemma}
\begin{proof}[Proof sketch]
Define $a_i(t)$ analogously to~\eqref{eigenexpansion}.  Paralleling the expansion of $\lambda_ia_i(t)$ in Lemma~\ref{discreteheatsolution} yields a \emph{second}-order ODE in $a_i(t)$ as a function of $t$, with no first-order term.  Applying the usual inhomogeneous second-order ODE solution formula and substituting the given boundary conditions yields the desired result.
\end{proof}

\subsection{Analysis and Properties}

The generic proof of the weak maximum principle in Lemma~\ref{weakmaximumprinciple} demonstrates how straightforward it can be to extend standard analytical proofs to the semi-discrete case.  Such structure preservation is a common theme in the treatment of discrete PDEs on graphs.  We provide two additional examples of continuous theorems that naturally extend to the semi-discrete case, one for the heat equation and one for the wave equation.

\begin{lemma}[\cite{chung07}, ``Huygens Property'' Theorem 3.6]
Suppose $u$ satisfies the heat equation $u_t=-\Delta u$ on a graph with $S=V$.  Then, for every $t,\delta>0$ we have
\begin{equation}
u(x,t+\delta)=\langle K(x,\cdot,\delta), u(\cdot,t)\rangle
\end{equation}
\end{lemma}
\begin{proof}
It is easy to check that the solution $u$ in this case satisfies
\begin{equation}\label{noboundarysoln}
u(x,t) = \frac{1}{\sqrt{d_x}}\sum_i c_i e^{-\lambda_it}\phi_i(x)
\end{equation}
for constants $\{c_i\}_{i=1,\ldots,|V|}$.  Thus,
\begin{align*}
u(x,t+\delta)
&= \frac{1}{\sqrt{d_x}} \sum_i c_ie^{-\lambda_i(t+\delta)}\phi_i(x)\textrm{ by~\eqref{noboundarysoln}}\\
&=\frac{1}{\sqrt{d_x}} \sum_i e^{-\lambda_i\delta}c_ie^{-\lambda_it}\phi_i(x)\\
&=\frac{1}{\sqrt{d_x}} \sum_i e^{-\lambda_i\delta}\phi_i(x)\left\langle \sqrt{d_y}\phi_i(y),\frac{1}{\sqrt{d_y}}\sum_j c_j e^{-\lambda_jt}\phi_j(y) \right\rangle\textrm{ by orthogonality of the $\phi_i$'s}\\
&=\frac{1}{\sqrt{d_x}} \sum_i e^{-\lambda_i\delta}\phi_i(x)\langle \sqrt{d_y}\phi_i(y), u(\cdot,t)\rangle\textrm{ by~\eqref{noboundarysoln} again}\\
&=\langle K(x,\cdot,\delta), u(\cdot,t)\rangle\textrm{, as desired}
\end{align*}
\end{proof}

Our next property provides a discrete analog of~\eqref{geometricwaveenergy}, which is also conserved over time.

\begin{lemma}[\cite{chung07}, Theorem 4.6]\label{energylemma}
Suppose $\delta S\neq\emptyset$ and $u$ satisfies the $\omega$-elastic equation with $u|_{\delta S}\equiv0$.  Then, the following energy function remains constant with respect to $t$:
\begin{equation}
E(t)=\int_{\bar{S}} \left[ u_t^2 + \frac{1}{2}\sum_{\substack{(x,y)\in E,y\in\bar{S}}} (D_yu(x,t))^2 \right]d_x\ d\V_x
\end{equation}
\end{lemma}
\begin{proof}
We show that the derivative of $E$ with respect to $t$ is zero:
\begin{align*}
E'(t) &= \int_{\bar{S}} \left[2u_tu_{tt}+\sum_y (D_yu\cdot D_yu_t) \right] d_x\ d\V_x\textrm{, differentiating under the (discrete) integral}\\
&= \int_{\bar{S}} \left[2u_tu_{tt}+\nabla u\cdot\nabla u_t \right] d_x\ d\V_x\textrm{ by definition of the discrete gradient}\\
&=2\int_{\bar{S}} u_t(u_{tt}+\Delta u)d_x\ d\V\textrm{ by discrete integration by parts~\cite{chung05}}\\
&=2\int_Su_t(u_{tt}+\Delta u)d_x\ d\V\textrm{ since we specified }u|_{\delta S}\equiv0\\
&=0\textrm{ since $u$ satisfies the wave equation $u_{tt}=-\Delta u$ in $S$}
\end{align*}
Thus, $E$ does not change as $t$ progresses.
\end{proof}
Note that this lemma implies that solutions of the discrete wave equation with prescribed boundary values are unique.  This is because the difference $u_1-u_2$ satisfies the conditions of Lemma~\ref{energylemma} with $E(0)=0$.  Thus, $\nicefrac{\partial}{\partial t}(u_1-u_2)=0$, and they agree at $t=0$.

\section{Conclusion}

The abstract definition of a graph as a collection of vertices and edges admits straightforward visual and physical \emph{intuition} but by design does not have a single application or construction in mind. Such generality is exactly what makes graphs versatile data structures, but at the same time provides limited guidance for the best adaptations of potentially valuable differential techniques.  Thus, there remains little consensus regarding the most effective methods for modeling flows, waves, and other phenomena as they might propagate along a graph.

Key to the development of a PDE theory for functions on graphs is the choice of topological domains and their accompanying operators.  Arguably closest to the classical case, \cite{friedman04} argues that the geometric realization $\G$ is the most natural domain.  Geometrizations have the advantage that they can be realized physically and have differential structure in edge interiors, allowing for the application of classical one-dimensional results with boundary conditions coupling vertex values.  Unfortunately, the transition from the manifold structure of subsets of $\R^1$ to the ``varifold'' structure of a graph presents considerable theoretical challenges:  Even the choice of Laplacian operators is unclear and may require combining two operators like $\Delta_E$ and $\Delta_V$ using integrating factors.  While abstract theorems characterizing graph diameters and other large scale properties can be proven in some cases, much remains unknown about $\G$ and whether it can be used to construct practical algorithms or analytical tools.

More practically implementable is the collection of techniques involving discrete differential operators on functions in $\R^{|V|}$.  These operators typically are expressible using sparse matrices that easily can be computed and manipulated.  In particular, depending on application there are at least three discrete Laplacian operators that might be useful, notated here as $L$, $\Delta$, and $\LL$.  Even more operators with similar structure can arise when dealing with Laplacians from geometry or simulation applications such as~\cite{hirani03}.  Thankfully, although the resulting flows can be different numerically, their broad properties remain the same and even can be simplified from the continuous case.  For instance, discrete Green's functions are constructed using matrix inversion, sidestepping the need for specialized convolution constructions and arguments about convergence~\cite{chung00}.

Partway between the fully continuous and fully discrete approaches are those methods allowing for semi-discrete treatment of PDEs with continuity in time but discrete values along graphs~\cite{chung07}.  These techniques take advantage of the fact that per-vertex functions on graphs admit some amount of ``discrete differential'' structure relating values at neighboring vertices, while allowing continuity in variables that are not naturally discretized by the graph itself such as time $t$.  These methods provide a potential compromise in which we can understand flows with continuous time, although computationally timestepping and eigenfunction computations will make related analytical techniques approximate.

Regardless of the domain and corresponding operator, many common themes arise when considering flows on graphs; most of these themes are direct adaptations of concepts from classical continuous theory.  Most importantly, all three of~\cite{friedman04,chung00,chung07} make use of operator eigenfunctions to construct closed-form solutions to model PDEs.  In the case of~\cite{friedman04}, these eigenfunctions exist in theory but are countably infinite, whereas in~\cite{chung00,chung07} they are finite.  Regardless, the theory of elliptic or positive definite operators guarantees that they fully characterize the behavior of the associated PDEs, and the straightforward form of the heat and wave equations simplifies their usage considerably.  Eigenvector computation is a well-studied technique in numerical analysis, so such closed-form solutions can be used directly in applications studying graphs with discrete operators.

A less technical common theme is that many of the qualitative properties of model PDEs are carried over from one domain to the other with little to no adaptation.  Many theorems characterizing solutions to differential equations can be expressed with little more than the list of properties of Laplacians in \S\ref{laplacianproperties}.  This observation confirms the intuition that the heat and wave equations have predictable properties on graphs even if they are not directly realizable in Euclidean space.

We have focused here on theoretical treatments of PDEs on graphs.  Work in this domain is not necessarily focused on providing algorithmic tools but can be applied directly and indirectly to formulate methods for analyzing networks.  For instance, having verified that the discrete graph Laplacian encodes aspects of structure in a similar way to smooth Laplacians indicates that geometric methods like~\cite{sun09} can be used for graph matching.  More generally, the propagation of heat and waves from vertex to vertex can help characterize graph topology at multiple scales as $t\rightarrow\infty$ through the use of straightforward numerical routines rather than specialized discrete algorithms.  Note that drawbacks similar to~\cite{gordon92} still apply, in that heat and wave kernel methods have no hope of distinguishing isospectral graphs.

Considerable work remains to be done toward understanding the advantages and drawbacks of network-based models of flows and differential equations.  Even so, the promising initial work of~\cite{friedman04,chung00,chung07} indicates that achieving deep results about these models may be as simple as adapting pre-existing proofs and techniques from the continuous domain.

\bibliographystyle{eg-alpha-doi}

\bibliography{qual}

\end{document}